\documentclass[
copyright,creativecommons]{eptcs}
\providecommand{\event}{GandALF 2026} 

\usepackage{iftex}

\ifpdf
  \usepackage{underscore}         
  \usepackage[T1]{fontenc}        
\else
  \usepackage{breakurl}           
\fi

\usepackage{amsmath,amssymb,amsthm}

\usepackage{longtable}
\usepackage{tabularx}
\usepackage{booktabs}
\usepackage{arydshln}
\usepackage{diagbox}

\usepackage{graphicx}
\usepackage{tikz}
\usetikzlibrary{arrows,automata,positioning}

\tikzset{
  position/.style args={#1:#2 from #3}{
    at=(#3.#1),
    anchor=#1+180,
    shift=(#1:#2)
  }
}

\usepackage{algorithm}
\usepackage{algpseudocode} 

\usepackage{enumitem}
\usepackage{multicol}

\usepackage{lineno}

\usepackage{xcolor}
\usepackage{hyperref}
\usepackage[nameinlink,noabbrev]{cleveref}

\hypersetup{
  colorlinks=true,
  linkcolor=blue,
  citecolor=blue,
  urlcolor=blue
}

\theoremstyle{plain}
\newtheorem{theorem}{Theorem}[section]

\newtheorem{fact}[theorem]{Fact}

\theoremstyle{definition}

\newtheorem{definition}[theorem]{Definition}

\newtheorem{example}[theorem]{Example}

\usepackage[numbers]{natbib}

\usepackage{xspace}

\newcommand{\Flg}{\mathcal{G}}
\newcommand{\Fln}{\mathcal{N}}

\newcommand{\FGCGF}{\mathtt{GCGF}}

\newcommand{\FXX}{\mathtt{X}}

\newcommand{\FCL}{\mathsf{CL}}

\newcommand{\FAG}{\mathtt{AG}}
\newcommand{\FAC}{\mathtt{AC}}
\newcommand{\FST}{\mathtt{ST}}

\newcommand{\Fav}{\mathtt{av}}  

\newcommand{\Fout}{\mathsf{out}}

\newcommand{\FJA}{\mathtt{JA}}

\newcommand{\FCC}{\mathrm{C}}

\newcommand{\FDD}{\mathrm{D}}

\newcommand{\FES}{\mathrm{ES}}

\newcommand{\Fdefs}[1]{\textbf{#1}}

\newcommand{\putaway}[1]{}

\title{Representation Theorems for Actual Powers \\ over Two-Agent General Concurrent Game Frames}

\author{
Zixuan Chen
\institute{Institute for Logic, Language and Computation, University of Amsterdam, Amsterdam, The Netherlands}
\email{zixuan.chen21@outlook.com}
\and
Fengkui Ju\thanks{Corresponding author}
\institute{School of Philosophy, Beijing Normal University, Beijing, China}
\email{fengkui.ju@bnu.edu.cn}
\and
Thomas {\AA}gotnes
\institute{Department of Information Science and Media Studies, University of Bergen, Bergen, Norway}
\institute{School of Philosophy, Shanxi University, Taiyuan, China}
\email{thomas.agotnes@uib.no}
}

\newcommand{\titlerunning}{Representation Theorems for Actual Powers}
\newcommand{\authorrunning}{Z. Chen, F. Ju \& T. {\AA}gotnes}

\hypersetup{
  bookmarksnumbered,
  pdftitle    = {\titlerunning},
  pdfauthor   = {\authorrunning},
  pdfsubject  = {EPTCS},               
  pdfkeywords = {general concurrent game frames, actual power, neighborhood frames, representation theorem} 
}

\begin{document}

\maketitle

\begin{abstract}
One of the most well-known connections between logic and games is Pauly's representation theorem: that the induced powers of individuals and coalitions in concurrent game frames correspond, in a precise sense, to a certain class of neighborhood frames. The precise sense here is what is called \emph{alpha effectivity} (or \emph{alpha power}): the power of a coalition is characterized by the sets of states which it can ensure the outcome to lie in by taking some joint action. This definition is inherently monotonic, and, as pointed out by \cite{benthem-new-2019}, that fact can obscure relevant information about the power structure in the game: we don't know whether two sets a coalition has the power to enforce correspond to the same or different joint actions. An alternative is to characterise the power of a coalition by its \emph{actual power} (called \emph{basic power} in \cite{benthem-new-2019}): the set of sets of states where each corresponds to one joint action by the coalition and all possible joint actions by the other agents.
It has recently been argued \cite{li-minimal-2025, li-completeness-2026} that standard concurrent game frames rely on three assumptions that in some cases may be too strong: seriality, independence of agents, and determinism. This gives a total of eight different classes of \emph{general} concurrent game frames.
In this paper, assuming two agents, we prove that for actual powers, the eight classes of general concurrent game frames are represented by eight corresponding classes of neighborhood frames.
This generalises a result in \cite{benthem-new-2019}.

\end{abstract}

\section{Introduction}


\emph{Concurrent game frames} are a standard semantic framework for logics of strategic reasoning. Many influential logics in this area are interpreted over concurrent game frames, including Coalition Logic $\mathsf{CL}$ \cite{pauly-modal-2002} and Alternating-time Temporal Logic $\mathsf{ATL}$ \cite{alur-alternating-time-2002}. Roughly speaking, a concurrent game frame consists of states and agents (forming coalitions), where, at each state, every coalition has available joint actions, and each such action induces a set of possible outcome states.

Two notions of coalition power arise from concurrent game frames: \emph{alpha powers} and \emph{actual powers}. An \emph{alpha power} of a coalition is a set of possible futures such that the coalition has an action that forces every resulting outcome to lie in that set. Alpha power is also known as alpha effectivity. An \emph{actual power} of a coalition is a set of possible futures such that the coalition has an action satisfying both of the following conditions: (1) the action forces every resulting outcome to lie in the set, and (2) every future state in the set is realizable by the complement coalition (and the environment). Thus, unlike alpha powers, actual powers explicitly involve agents outside the coalition. Note that actual powers contain more information about the power structure than alpha powers: we can get the latter from the former, but not the other way around.

Correspondingly, one obtains two kinds of neighborhood frames: \emph{alpha neighborhood frames} and \emph{actual neighborhood frames}. In an alpha neighborhood frame, each coalition is assigned an \emph{alpha neighborhood function} that, at every state, specifies an \emph{alpha neighborhood}, i.e., a set of alpha powers. In an actual neighborhood frame, each coalition is assigned an \emph{actual neighborhood function} that, at every state, specifies an \emph{actual neighborhood}, i.e., a set of actual powers.

\emph{Alpha powers} are widely used in game theory; see, e.g.,~\cite{moulin-cores-1982}. 
By contrast, there appears to be comparatively little work on \emph{actual} powers. 
To the best of our knowledge, the main reference is van Benthem, Bezhanishvili, and Enqvist~\cite{benthem-new-2019}. 
Working primarily from the perspective of game equivalence, they investigate \emph{basic} powers of singleton coalitions in turn-based two-agent extensive games with imperfect information, defined in terms of (uniform) strategies. 
Conceptually, this notion of basic power coincides with what we call \emph{actual} power.

Alpha and actual powers abstract away from the underlying actions in different
ways.
Alpha powers record guarantees: if a coalition has an action whose possible
outcomes are contained in a set \(X\), then the same action also guarantees any
superset of \(X\). Hence alpha powers are upward closed.
Actual powers, by contrast, record the precise outcome set generated by an
available coalition action: the set of states left possible once that action
is fixed and the choices of the remaining agents are left open. They are
therefore generally not upward closed.
As observed in~\cite{benthem-new-2019}, describing a game solely in terms of
its alpha powers can therefore lose information about its power structure,
whereas actual powers provide a finer-grained account.

Because actual powers retain the outcomes left open to agents outside the
coalition, they are useful for modeling social scenarios in which a
coalition's ability depends not only on what it can guarantee, but also on what
its action leaves open to others. A central logical example is Socially
Friendly Coalition Logic~\cite{goranko-socially-2018}. Its operator
$[\FCC](\phi;\psi_1,\dots,\psi_k)$ expresses that coalition $\FCC$ has a
collective action $\sigma_\FCC$ that guarantees $\phi$ while allowing the
complementary coalition $\overline{\FCC}$ to realize each of the alternatives
$\psi_1,\dots,\psi_k$ by a suitable response. This semantic clause is the
coalition-indexed, action-based counterpart of the instantial neighborhood
operator
$\Box(\phi;\psi_1,\dots,\psi_k)$ of~\cite{benthem-instantial-2017}.
Equivalently, in actual-neighborhood terms, there must be an actual power $X$
of $\FCC$ such that every state in $X$ satisfies $\phi$ and, for every
$i\leq k$, $X$ contains some state satisfying $\psi_i$. The latter requirement
depends on the precise outcome range and cannot be captured by alpha powers
alone.


\medskip

Every concurrent game frame induces an alpha neighborhood frame.
Pauly~\cite{pauly-modal-2002} and Goranko, Jamroga, and Turrini~\cite{goranko-strategic-2013} showed that the class of concurrent game frames is \emph{represented} by a class of alpha neighborhood frames characterized by certain \emph{basic} properties of neighborhood functions, in the following sense:
(1) every alpha neighborhood frame induced by a concurrent game frame satisfies these properties, and
(2) every alpha neighborhood frame satisfying these properties is induced by some concurrent game frame.
This result is the \emph{alpha representation theorem} for concurrent game frames.

Every concurrent game frame also induces an actual neighborhood frame. Do there exist \emph{basic} properties of neighborhood functions such that
(1) every actual neighborhood frame induced by a concurrent game frame satisfies these properties, and
(2) every actual neighborhood frame satisfying these properties is induced by some concurrent game frame?
If so, this would be an \emph{actual representation theorem} for concurrent game frames.

This question has not yet been fully settled in the literature. As far as we are aware, the only closely related result is due to van Benthem, Bezhanishvili, and Enqvist~\cite{benthem-new-2019}. There, the authors establish a representation theorem for \emph{basic} (strategy-based) powers in turn-based two-agent extensive games with imperfect information, under the restriction to singleton coalitions. Via a standard power-invariance transformation between the class of two-agent concurrent game frames and an appropriate class of turn-based two-agent extensive games with imperfect information, their proof readily yields a representation theorem for actual powers in two-agent concurrent game frames, again restricted to singleton coalitions (see ~\cite{chen-representation-2026} for details).

Let us briefly discuss the significance of representation theorems for alpha and actual powers.
What are the powers of players, and coalitions of players, in games? Which sets of outcomes -- known as \emph{events} in probability theory or \emph{propositions} in logic -- is it in their power to make come about, in the sense that they can choose a (joint) action such that no matter what the other agents do, this is what will happen? This abstracts away the particulars of a game, and leaves an abstract representation of power which is exactly what we are often interested in when we analyse games. In order to understand power in games, the central question then is: what are the \emph{properties} (or \emph{axioms} or \emph{laws}) of individual and coalitional power induced from games in this way? That is exactly what Pauly's representation theorem \cite{pauly-modal-2002} answers, for alpha powers.
This is the main motivation, and the main significance of also the representation theorems for actual power: giving us a deeper understanding of the power in games.

Representation theorems can also yield alternative neighborhood semantics for some
logics originally interpreted over concurrent game models. This makes the whole toolchest developed to deal with neighbourhood semantics available. Such semantics are also 
often more economical and convenient. Whereas a concurrent game frame
explicitly records action names, availability functions, and outcome functions,
many modal languages for strategic reasoning depend only on the powers induced
by actions: what coalitions can force in the case of alpha powers, or the exact
outcome ranges generated by coalition actions in the case of actual powers.
Representation theorems therefore justify replacing action-based models with
neighborhood-based models, without changing truth conditions.
Pauly's representation theorem for alpha powers provides a prominent example:
it allows the semantics of $\FCL$ over concurrent game frames to be transferred
to alpha neighborhood frames. Pauly~\cite{pauly-modal-2002} used such frames to
prove completeness of $\FCL$, while {\AA}gotnes and
Alechina~\cite{agotnes-coalition-2019} employed the same semantic framework to
extend $\FCL$ with knowledge.


\medskip

Li and Ju~\cite{li-minimal-2025,li-completeness-2026} argued that standard
concurrent game frames impose three assumptions that may be too strong in
practice: \emph{seriality}, \emph{independence of agents}, and
\emph{determinism}. For motivating examples, see~\cite{li-minimal-2025}. They
introduced eight classes of general concurrent game frames, corresponding to
the eight possible combinations of these assumptions, and studied the
generalized coalition logic associated with each class.

In this paper, assuming two agents, we establish the following result:
for actual powers, the eight classes of general concurrent game frames are represented by eight corresponding classes of actual neighborhood frames.
Our proof method relies essentially on the two-agent setting, and does not extend directly to settings with three or more agents. We briefly explain this at the end of Section \ref{section: Representing actual powers}.

The remainder of the paper is organized as follows.
In Section~\ref{section: General settings}, we introduce the general framework
used throughout the paper, including action frames, alpha and actual powers,
actual neighborhood frames, and representability of action frames by actual
neighborhood frames.
In Section~\ref{section: Eight classes of general concurrent game frames}, we define the eight classes of general concurrent game frames.
In Section~\ref{section: Representing actual powers}, we define eight classes of actual neighborhood frames via eight finite sets of properties of actual neighborhood functions, and show that, for actual powers, the eight classes of general concurrent game frames (with two agents) are represented by the corresponding eight classes of actual neighborhood frames.
In Section~\ref{section: Concluding remarks}, we conclude the paper with final remarks.


Due to space limitations, we omit some material that can be found in the full paper~\cite{chen-representation-2026}. Moreover, building on the representation result for actual powers established in this paper, the full paper shows that, for alpha powers, the same eight classes are represented by eight corresponding classes of alpha neighborhood frames.

\newcommand{\FAF}{\mathtt{AF}}

\section{Preliminaries}
\label{section: General settings}

In this section, we proceed as follows. First, we introduce \emph{action
frames} in an abstract form; specific constraints will be imposed later.
Second, we discuss \emph{alpha powers} and \emph{actual powers}, and briefly
compare them. Third, we present \emph{actual neighborhood frames}. These are
again introduced at an abstract level, with additional constraints to be added
in subsequent sections. Finally, we define how action frames are represented
by actual neighborhood frames.

\subsection{Action frames}

Let $\FAG$ be a finite nonempty set of agents. Subsets of $\FAG$ are called \Fdefs{coalitions}, and $\FAG$ itself is called the \Fdefs{grand coalition}.
In what follows, whenever no confusion arises, we write $a$ instead of $\{a\}$ for $a \in \FAG$.

Let $\FAC$ be a nonempty set of actions. For each $\FCC \subseteq \FAG$, define
$\FJA_\FCC = \{\sigma_\FCC \mid \sigma_\FCC: \FCC \rightarrow \FAC\}$,
the set of \Fdefs{joint actions} of $\FCC$. Note $\FJA_\emptyset =\{\emptyset\}$. For $\FCC, \FDD \subseteq \FAG$ such that $\FCC \subseteq \FDD$, and for every $\sigma_\FDD \in \FJA_\FDD$, we denote by $\sigma_\FDD|_\FCC$ the \Fdefs{restriction} of $\sigma_\FDD$ to $\FCC$, which is an element of $\FJA_\FCC$.
In the sequel, we sometimes represent joint actions of coalitions as action sequences, implicitly assuming a fixed order of agents.

\begin{definition}[Action frames]

An \Fdefs{action frame} is a tuple
\[
\FAF = (\FST, \FAC, \{\Fav_\FCC \mid \FCC \subseteq \FAG\}, \{\Fout_\FCC \mid \FCC \subseteq \FAG\}),
\]
where:

\begin{itemize}

\item
$\FST$ is a nonempty set of states.

\item
$\FAC$ is a nonempty set of actions.

\item
For each $\FCC \subseteq \FAG$, $\Fav_\FCC: \FST \rightarrow \mathcal{P}(\FJA_\FCC)$ is an \Fdefs{availability function} for $\FCC$.

\emph{Here, $\Fav_\FCC(s)$ is the set of all joint actions of coalition $\FCC$ that are available at state $s$.}

\item
For each $\FCC \subseteq \FAG$, $\Fout_\FCC: \FST \times \FJA_\FCC \rightarrow \mathcal{P}(\FST)$ is an \Fdefs{outcome function} for $\FCC$.

\emph{Here, $\Fout_\FCC(s, \sigma_\FCC)$ is the set of possible outcome states when coalition $\FCC$ performs joint action $\sigma_\FCC$ at state $s$.}

\end{itemize}

\end{definition}

Note that the action set $\FAC$ may vary across action frames.

Different constraints may be imposed on action frames.

Intuitively, $\Fout_\emptyset(s,\emptyset)$ contains all states to which $s$ may evolve, i.e., all possible successors of $s$.

\subsection{Alpha and actual powers}

Fix an action frame
\[
\FAF = (\FST, \FAC, \{\Fav_\FCC \mid \FCC \subseteq \FAG\}, \{\Fout_\FCC \mid \FCC \subseteq \FAG\}),
\]
a state $s \in \FST$, a coalition $\FCC \subseteq \FAG$, and an available joint action $\sigma_\FCC \in \Fav_\FCC(s)$.
We say that a set of states $X \subseteq \FST$ is \Fdefs{safe} for $\sigma_\FCC$ at $s$ if
$\Fout_\FCC(s,\sigma_\FCC)\subseteq X$,
that is, performing $\sigma_\FCC$ at $s$ guarantees that the next state lies in $X$.
We say that $X$ is \Fdefs{tight} for $\sigma_\FCC$ at $s$ if
$X\subseteq \Fout_\FCC(s,\sigma_\FCC)$,
that is, every state in $X$ can occur as an outcome compatible with performing $\sigma_\FCC$ at $s$.

An \Fdefs{alpha power} of $\FCC$ at $s$ is a set of states $X \subseteq \FST$ that is safe for some available joint action of $\FCC$ at $s$.
An \Fdefs{actual power} of $\FCC$ at $s$ is a set of states $X \subseteq \FST$ that is both safe and tight for some available joint action of $\FCC$ at $s$. Equivalently, actual powers of $\FCC$ at $s$ are exactly the sets of the form $\Fout_\FCC(s,\sigma_\FCC)$ for $\sigma_\FCC \in \Fav_\FCC(s)$.
We say that an (alpha or actual) power $X$ \Fdefs{enables} a state $t \in \FST$ at $s$ if $t \in X$.

Alpha powers are monotonic: if $X$ is an alpha power, then every superset $X' \supseteq X$ is also an alpha power.
By contrast, actual powers are generally non-monotonic, since tightness rules out adding states that are not actually reachable.

The following example illustrates the distinction between alpha and actual
powers.

\begin{example}
\label{ex:01-same-alpha-powers-different-actual-powers}

Consider two AI-driven robots, $a$ and $b$, controlling a warehouse gate. The
robots send their commands to the gate controller simultaneously, but robot
$a$ has higher authority than robot $b$.

There are two relevant states, \(W=\{w_1,w_2\}\), where $w_1$ represents the
gate being closed and $w_2$ represents the gate being open. Robot $a$ has three
actions:
\(\mathtt{open}_a\), \(\mathtt{close}_a\), and
\(\mathtt{delegate}_a\). The first opens the gate, the second closes it, and
the third delegates the decision to robot $b$. Robot $b$ has two actions:
\(\mathtt{open}_b\) and \(\mathtt{close}_b\).

We consider two scenarios. In both, if $a$ chooses \(\mathtt{open}_a\), the
gate opens, and if $a$ chooses \(\mathtt{close}_a\), the gate closes,
regardless of $b$'s action. The scenarios differ only in the effect of
\(\mathtt{delegate}_a\). In the first scenario, delegation works as intended:
when $a$ delegates, $b$'s action determines whether the gate opens or closes.
In the second scenario, delegation is ineffective: when $a$ delegates, the
gate closes regardless of $b$'s action.

The two scenarios induce the same alpha powers for every coalition:
\[
\begin{array}{c|c}
\text{coalition} & \text{alpha powers} \\
\hline
\emptyset & \{W\} \\
\{a\} & \bigl\{\{w_1\},\{w_2\},W\bigr\} \\
\{b\} & \{W\} \\
\{a,b\} & \bigl\{\{w_1\},\{w_2\},W\bigr\}.
\end{array}
\]

Their actual powers differ. In the first scenario, the actual
powers of \(\{a\}\) are
\(\bigl\{\{w_1\},\{w_2\},W\bigr\}\), corresponding to closing, opening, and
delegating, respectively. In the second scenario, they are only
\(\bigl\{\{w_1\},\{w_2\}\bigr\}\). Indeed, delegation no longer leaves both
outcomes open, but instead has the same exact outcome set as
\(\mathtt{close}_a\).

\end{example}

\newcommand{\FACNF}{\mathtt{NF}^\mathtt{ac}}
\newcommand{\FALNF}{\mathtt{NF}^\alpha}

\newcommand{\Facnei}{\mathrm{N}^\mathrm{ac}}
\newcommand{\Falnei}{\mathrm{N}^\alpha}

\newcommand{\Falneinc}{\mathrm{CoreN}^\alpha}

\subsection{Actual neighborhood frames}

\begin{definition}[Actual neighborhood frames]

An \Fdefs{actual neighborhood frame} is a tuple
\[
\FACNF = (\FST, \{\Facnei_\FCC \mid \FCC \subseteq \FAG\}),
\]
where:

\begin{itemize}

\item
$\FST$ is a nonempty set of states;

\item
For each $\FCC \subseteq \FAG$, $\Facnei_\FCC: \FST \to \mathcal{P}(\mathcal{P}(\FST))$ is an \Fdefs{actual neighborhood function} of $\FCC$.

\end{itemize}

\end{definition}

Different constraints may be imposed on actual neighborhood frames.
The set $\Facnei_\FCC(s)$ is called the \Fdefs{actual neighborhood} of $\FCC$ at $s$. Intuitively, $\Facnei_\FCC(s)$ consists of the actual powers of coalition $\FCC$ at state $s$.
In the intended interpretation, $\bigcup \Facnei_\emptyset(s)$ consists of all states to which $s$ may evolve, i.e., the possible successors of $s$.

Structurally, actual neighborhood frames are ordinary coalition-indexed neighborhood frames. The qualifier “actual” merely emphasizes that their neighborhoods are interpreted as representing the actual powers of coalitions.

\newcommand{\ab}{\allowbreak}

\newcommand{\FALEF}{\mathrm{EF}^\alpha}
\newcommand{\FACEF}{\mathrm{EF}^\mathrm{ac}}
\newcommand{\PAL}{\alpha}

\newcommand{\PAC}{\mathtt{AC}}

\subsection{Representation of action frames by actual neighborhood frames}

\begin{definition}[Actual effectivity functions of action frames]

Let $\FAF = (\FST, \FAC, \{\Fav_\FCC \mid \FCC \subseteq \FAG\}, \{\Fout_\FCC \mid \FCC \subseteq \FAG\})$ be an action frame.
For each $\FCC \subseteq \FAG$, define:

\begin{itemize}

\item
the \Fdefs{actual effectivity function} $\FACEF_\FCC$ for $\FCC$ in $\FAF$ by setting, for every $s \in \FST$,
\[
\FACEF_\FCC(s)
=
\{\, \Fout_\FCC(s,\sigma_\FCC) \mid \sigma_\FCC \in \Fav_\FCC(s) \,\}.
\]

\end{itemize}

\end{definition}

\begin{definition}[$\PAC$-representability of action frames]

Let $\FAF = (\FST, \FAC, \ab \{\Fav_\FCC \mid \FCC \subseteq \FAG\}, \{\Fout_\FCC \mid \FCC \subseteq \FAG\})$ be an action frame.

We say that $\FAF$ is \Fdefs{$\PAC$-represented} by an actual neighborhood frame
\[
\FACNF = (\FST, \{\Facnei_\FCC \mid \FCC \subseteq \FAG\})
\]
if, for every $\FCC \subseteq \FAG$, we have $\FACEF_\FCC = \Facnei_\FCC$.

\end{definition}

Intuitively, $\PAC$-representability means that the actual neighborhood frame
captures all information about the actual powers of coalitions in the action
frame.

\begin{definition}[$\PAC$-representability of classes of action frames]

Let $\mathbf{AF}$ be a class of action frames.
We say that $\mathbf{AF}$ is \Fdefs{$\PAC$-represented} by a class of actual
neighborhood frames $\mathbf{AC}\text{-}\mathbf{NF}$ if:

\begin{itemize}

\item
Every action frame in $\mathbf{AF}$ is $\PAC$-represented by some actual
neighborhood frame in $\mathbf{AC}\text{-}\mathbf{NF}$;

\item
Every actual neighborhood frame in $\mathbf{AC}\text{-}\mathbf{NF}$
$\PAC$-represents some action frame in $\mathbf{AF}$.

\end{itemize}

\end{definition}

Suppose that a class of action frames is $\PAC$-represented by a class of
actual neighborhood frames. Intuitively, this means that, as far as actual
powers are concerned, one can translate the information carried by action
frames into the corresponding neighborhood frames.

Now fix a class of action frames $\mathbf{AF}$ and a collection of properties
of actual neighborhood frames. To show that $\mathbf{AF}$ is
$\PAC$-represented by the class of actual neighborhood frames satisfying
these properties, it suffices to establish the following two statements:

\begin{itemize}

\item
\Fdefs{Havingness} (necessity): every actual neighborhood frame that
$\PAC$-represents an action frame in $\mathbf{AF}$ satisfies these properties;

\item
\Fdefs{Enoughness} (sufficiency): every actual neighborhood frame satisfying
the properties $\PAC$-represents an action frame in $\mathbf{AF}$.

\end{itemize}

Later we will follow exactly this pattern.

\putaway{
\paragraph{Remarks.}

Fix a class of action frames $\mathbf{AF}$. From each $\FAF \in \mathbf{AF}$, we obtain an alpha (resp.\ actual) neighborhood frame $\FALNF$ (resp.\ $\FACNF$) that $\PAL$-represents (resp.\ $\PAC$-represents) $\FAF$.
Let $\alpha\text{-}\mathbf{NF}$ (resp.\ $\mathbf{AC}\text{-}\mathbf{NF}$) be the class of all such alpha (resp.\ actual) neighborhood frames.
Then, trivially, $\mathbf{AF}$ is $\PAL$-represented (resp.\ $\PAC$-represented) by $\alpha\text{-}\mathbf{NF}$ (resp.\ $\mathbf{AC}\text{-}\mathbf{NF}$).
This observation is not very informative: to obtain a meaningful representation, the classes $\alpha\text{-}\mathbf{NF}$ and $\mathbf{AC}\text{-}\mathbf{NF}$ should be specified by some \emph{good} properties of alpha and actual neighborhood frames rather than by construction from $\mathbf{AF}$. Of course, \emph{good} is a vague notion. 
}

\section{General concurrent game frames and eight classes}
\label{section: Eight classes of general concurrent game frames}

In this section, we present the eight classes of general concurrent game frames considered by Li and Ju~\cite{li-minimal-2025, li-completeness-2026}. Later, assuming two agents, we establish their $\PAC$-representability.

\begin{definition}[General concurrent game frames]

An action frame
\[
\FGCGF = (\FST, \FAC, \{\Fav_\FCC \mid \FCC \subseteq \FAG\}, \{\Fout_\FCC \mid \FCC \subseteq \FAG\})
\]
is a \Fdefs{general concurrent game frame} if the following conditions hold:

\begin{itemize}

\item

The \Fdefs{grand-coalition-induced outcome condition} (simply the \Fdefs{GCI-condition}):
for every $\FCC \subseteq \FAG$, every $s \in \FST$, and every $\sigma_\FCC \in \FJA_\FCC$,
\[
\Fout_\FCC(s,\sigma_\FCC)
=
\bigcup
\{\, \Fout_\FAG(s,\sigma_\FAG) \mid \sigma_\FAG \in \FJA_\FAG \text{ and } \sigma_\FAG|_\FCC = \sigma_\FCC \,\}.
\]

\item

The \Fdefs{outcome-driven availability condition} (simply the \Fdefs{ODA-condition}):
for every $\FCC \subseteq \FAG$ and every $s \in \FST$,
\[
\Fav_\FCC(s)
=
\{\, \sigma_\FCC \in \FJA_\FCC \mid \Fout_\FCC(s,\sigma_\FCC) \neq \emptyset \,\}.
\]

\end{itemize}

\end{definition}

Thus, under the GCI-condition and the ODA-condition, the functions $\Fout_\FCC$ and $\Fav_\FCC$ (for all $\FCC \subseteq \FAG$) are determined by the grand-coalition outcome function $\Fout_\FAG$.

\begin{definition}[Seriality, independence, and determinism]
\label{definition:Seriality, independence, and determinism of general concurrent game frames}

Let $\FGCGF = (\FST, \FAC, \{\Fav_\FCC \mid \FCC \subseteq \FAG\}, \{\Fout_\FCC \mid \FCC \subseteq \FAG\})$ be a general concurrent game frame. We say that:

\begin{itemize}

\item
$\FGCGF$ is \Fdefs{serial} if, for all $\FCC \subseteq \FAG$ and all $s \in \FST$, we have $\Fav_\FCC(s) \neq \emptyset$;

\item
$\FGCGF$ is \Fdefs{independent} if, for all $s \in \FST$ and all disjoint coalitions $\FCC,\FDD \subseteq \FAG$,
whenever $\sigma_\FCC \in \Fav_\FCC(s)$ and $\sigma_\FDD \in \Fav_\FDD(s)$, we have $\sigma_\FCC \cup \sigma_\FDD \in \Fav_{\FCC \cup \FDD}(s)$;

\item
$\FGCGF$ is \Fdefs{deterministic} if, for all $s \in \FST$ and all $\sigma_\FAG \in \Fav_\FAG(s)$, the set $\Fout_\FAG(s,\sigma_\FAG)$ is a singleton.

\end{itemize}

\end{definition}

Let $\mathtt{S}$, $\mathtt{I}$, and $\mathtt{D}$ denote seriality, independence, and determinism, respectively. Let
\[
\FES = \{\epsilon,\ \mathtt{S},\ \mathtt{I},\ \mathtt{D},\ \mathtt{SI},\ \mathtt{SD},\ \mathtt{ID},\ \mathtt{SID}\}
\]
be the set of the eight strings encoding all combinations of these three properties. Here $\epsilon$ denotes the empty combination.
For $\FXX \in \FES$, we say that a general concurrent game frame is an \Fdefs{$\FXX$-frame} if it satisfies the properties indicated by $\FXX$.

The following facts will be used implicitly throughout the sequel.

\begin{fact}
\label{fact:gcgf-basic-consequences}

Let
\[
\FGCGF = (\FST, \FAC, \{\Fav_\FCC \mid \FCC \subseteq \FAG\}, \{\Fout_\FCC \mid \FCC \subseteq \FAG\})
\]
be a general concurrent game frame. Then the following hold:

\begin{enumerate}

\item
\Fdefs{Outcome monotonicity:} for all $s \in \FST$, all $\FCC \subseteq \FDD \subseteq \FAG$, all $\sigma_\FCC \in \FJA_\FCC$, and all $\sigma_\FDD \in \FJA_\FDD$ such that $\sigma_\FDD|_\FCC = \sigma_\FCC$, we have
\[
\Fout_\FDD(s,\sigma_\FDD)\subseteq \Fout_\FCC(s,\sigma_\FCC).
\]

\item

\Fdefs{Alternative GCI-condition:}
for all $s \in \FST$, all $\FCC \subseteq \FAG$, and all $\sigma_\FCC \in \FJA_\FCC$,
\[
\Fout_\FCC(s,\sigma_\FCC)
=
\bigcup
\{\, \Fout_\FAG(s,\sigma_\FAG)\mid \sigma_\FAG \in \Fav_\FAG(s)\ \text{and}\ \sigma_\FAG|_\FCC=\sigma_\FCC \,\}.
\]

\end{enumerate}

\end{fact}

The first item is an immediate consequence of the GCI-condition, and the second is a reformulation of the GCI-condition using the ODA-condition (since unavailable grand-coalition actions have empty outcome sets).

\putaway{
\paragraph{Remarks.}

The main question addressed below is whether any, or all, of the eight classes
of general concurrent game frames are $\PAC$-represented by classes of actual
neighborhood frames specified directly by conditions on actual neighborhood
functions.
}

\section{Representing actual powers}
\label{section: Representing actual powers}

In this section, we proceed as follows. First, we define eight classes of $\PAC$-representative actual neighborhood frames and establish several basic facts about them. Second, under the assumption that there are exactly two agents, we show that each of the eight classes of general concurrent game frames introduced in the previous section is $\PAC$-represented by the corresponding class of $\PAC$-representative actual neighborhood frames. The representability result has two directions: \emph{havingness} and \emph{enoughness}.

\subsection{Eight classes of $\PAC$-representative actual neighborhood frames}
\label{subsection: Eight classes of AC-representative actual neighborhood frames}

The following definitions are stated for an arbitrary finite set of agents $\FAG$. Later, we prove the representation theorem under the assumption that $|\FAG|=2$ for the enoughness direction. The havingness direction does not require the two-agent restriction.

\begin{definition}[$\PAC$-representative actual neighborhood frames]
\label{definition: AC-representative actual neighborhood frames with two agents}

Let $\FACNF = (\FST, \{\Facnei_\FCC \mid \FCC \subseteq \FAG\})$ be an actual neighborhood frame.
We say that $\FACNF$ is \Fdefs{$\PAC$-representative} if the following conditions are satisfied:

\begin{enumerate}

\item \Fdefs{Actual triviality of the empty coalition:}
for all $s \in \FST$, if $\Facnei_\emptyset(s)$ is nonempty, then it is a singleton.

\emph{Intuitively, this condition says that the empty coalition has at most one actual power. Having more than one actual power would amount to having more than one choice, hence the name.}

\item \Fdefs{Liveness:}
for all $s \in \FST$ and all $\FCC \subseteq \FAG$, we have $\emptyset \notin \Facnei_\FCC(s)$.

\emph{Intuitively, this condition rules out the degenerate actual power $\emptyset$ for any coalition. Note that liveness does not force coalitions to have powers; it only rules out treating the empty set as a ``power'' when powers exist.}

\item \Fdefs{Actual power inclusion:}
for all $s \in \FST$ and all $\FCC, \FDD \subseteq \FAG$ with $\FCC \subseteq \FDD$, for all $X \in \Facnei_\FDD(s)$ there exists $Y \in \Facnei_\FCC(s)$ such that $X \subseteq Y$.

\emph{Intuitively, every actual power of a larger coalition is included in some actual power of a smaller coalition.}

\item \Fdefs{Actual power decomposition:}
for all $s \in \FST$ and all $\FCC, \FDD \subseteq \FAG$ with $\FCC \subseteq \FDD$, for every $X \in \Facnei_\FCC(s)$ there exists $\Delta \subseteq \Facnei_\FDD(s)$ such that $X = \bigcup \Delta$.

\emph{Intuitively, every actual power of a smaller coalition can be decomposed into (possibly many) actual powers of a larger coalition.}

\end{enumerate}

\end{definition}

We use ``$\PAC$-representative'' for actual neighborhood frames satisfying the intrinsic conditions intended to characterize those that $\PAC$-represent general concurrent game frames. Note ``$\PAC$-representative'' and ``$\PAC$-representable'' are defined in different ways.

\begin{definition}[Seriality, independence, and determinism of $\PAC$-representative actual neighborhood fra\-mes]
\label{definition:Seriality, independence, and determinism of AC-representative actual neighborhood frames}

Let $\FACNF = (\FST, \{\Facnei_\FCC \mid \FCC \subseteq \FAG\})$ be an $\PAC$-representative actual neighborhood frame.
We say that:

\begin{itemize}

\item $\FACNF$ is \Fdefs{$\PAC$-serial} if for all $s \in \FST$ and all $\FCC \subseteq \FAG$, we have $\Facnei_\FCC(s) \neq \emptyset$.

\item $\FACNF$ is \Fdefs{$\PAC$-independent} if, for all $s \in \FST$ and all $\FCC, \FDD \subseteq \FAG$ with $\FCC \cap \FDD = \emptyset$, whenever $X \in \Facnei_\FCC(s)$ and $Y \in \Facnei_\FDD(s)$, there exists $Z \in \Facnei_{\FCC \cup \FDD}(s)$ such that $Z \subseteq X \cap Y$.

\emph{Intuitively, for disjoint coalitions, any two actual powers can be jointly refined to an actual power of their union that is compatible with both.}

\item $\FACNF$ is \Fdefs{$\PAC$-deterministic} if for all $s \in \FST$ and all $X \in \Facnei_\FAG(s)$, $X$ is a singleton.

\end{itemize}

\end{definition}

As before, we let the three symbols $\mathtt{S}$, $\mathtt{I}$, and $\mathtt{D}$ denote $\PAC$-seriality, $\PAC$-independence, and $\PAC$-determinism, respectively. We also reuse the set
\[
\FES = \{\epsilon,\ \mathtt{S},\ \mathtt{I},\ \mathtt{D},\ \mathtt{SI},\ \mathtt{SD},\ \mathtt{ID},\ \mathtt{SID}\},
\]
whose elements serve as labels for the eight possible combinations of these three properties.

For each $\FXX\in\FES$ and each $\PAC$-representative actual neighborhood
frame $\FACNF$, we say that $\FACNF$ is an \Fdefs{$\FXX$-frame} if it satisfies
the properties corresponding to $\FXX$.

\medskip

What follows are several basic facts about $\PAC$-representative actual neighborhood frames. 
These facts are useful both for intuition and for later technical developments.
Again, we refer to \cite{chen-representation-2026} for their proofs.

The first fact shows that the set of states reachable from a given state is independent of the coalition:
all coalitions share the same set of successors of that state.

\begin{fact}
\label{fact: enabling the same successors}

Let $\FACNF = (\FST, \{\Facnei_\FCC \mid \FCC \subseteq \FAG\})$ be an $\PAC$-representative actual neighborhood frame.
For all $s\in \FST$ and all $\FCC \subseteq \FAG$, we have
\[
\bigcup \Facnei_\FCC(s) = \bigcup \Facnei_\emptyset (s).
\]

\end{fact}

The next fact shows that emptiness of the set of actual powers at a state does not depend on the coalition:
either every coalition has actual power at $s$, or none does.

\begin{fact}
\label{fact:either all empty or none empty}

Let $\FACNF = (\FST, \{\Facnei_\FCC \mid \FCC \subseteq \FAG\})$ be an $\PAC$-representative actual neighborhood frame.
For all $s\in \FST$ and all $\FCC \subseteq \FAG$,
\[
\Facnei_\FCC(s)=\emptyset \quad \text{iff} \quad \Facnei_\emptyset(s)=\emptyset.
\]

\end{fact}

The following fact shows that if the empty coalition has any actual power at a state $s$, then it has exactly one: namely, the overall successor set of~$s$.

\begin{fact}
\label{fact:empty-coalition-power}

Let $\FACNF = (\FST, \{\Facnei_\FCC \mid \FCC \subseteq \FAG\})$ be an $\PAC$-representative actual neighborhood frame.
For all $s\in \FST$, if $\Facnei_\emptyset (s)\neq \emptyset$, then
\[
\Facnei_\emptyset (s) \;=\; \bigl\{\, \bigcup \Facnei_\FAG (s) \,\bigr\}.
\]

\end{fact}

The final fact states that, when moving from a smaller coalition to a larger one, every actual power $X$ of the
smaller coalition can be recovered as the union of those actual powers of the larger coalition that are included in $X$.

\begin{fact}
\label{fact:X union Delta X}

Let $\FACNF = (\FST, \{\Facnei_\FCC \mid \FCC \subseteq \FAG\})$ be an $\PAC$-representative actual neighborhood frame.
For every $s \in \FST$, all $\FCC, \FDD \subseteq \FAG$ with $\FCC \subseteq \FDD$, and every $X \in \Facnei_\FCC(s)$,
\[
\bigcup \{Z \in \Facnei_\FDD (s) \mid Z \subseteq X\} = X.
\]

\end{fact}

\subsection{Representation theorems for actual powers}
\label{subsection: Actual havingness theorems}

\begin{theorem}[Actual havingness theorem]

Let $\FXX \in \FES$.
For every actual neighborhood frame, if it $\PAC$-represents a general concurrent game $\FXX$-frame, then it is an $\PAC$-representative actual neighborhood $\FXX$-frame.

\end{theorem}

\begin{proof}

Let $\FACNF = (\FST, \{\Facnei_\FCC \mid \FCC \subseteq \FAG\})$ be an actual neighborhood frame.
Suppose that $\FACNF$ $\PAC$-represents a general concurrent game $\FXX$-frame
\[
\FGCGF= (\FST, \FAC, \{\Fav_\FCC \mid \FCC \subseteq \FAG\}, \{\Fout_\FCC \mid \FCC \subseteq \FAG\}).
\]
Then for every coalition $\FCC \subseteq \FAG$ and every state $s\in \FST$, we have $\Facnei_\FCC(s)=\FACEF_\FCC(s)$.

\paragraph{We show that $\FACNF$ is an $\FXX$-frame.}

We verify that each of the properties $\mathtt{S}$, $\mathtt{I}$, and $\mathtt{D}$ is preserved under $\PAC$-representation. It then follows, since $\FGCGF$ is an $\FXX$-frame, that $\FACNF$ is an $\FXX$-frame.

\begin{itemize}

\item \textbf{Seriality.}
Assume that $\FGCGF$ is serial. Fix $s \in \FST$ and $\FCC \subseteq \FAG$.
Then $\Fav_\FCC(s)\neq\emptyset$, and hence
\[
\FACEF_\FCC(s)
=\{\Fout_\FCC(s,\sigma_\FCC)\mid \sigma_\FCC\in\Fav_\FCC(s)\}\neq\emptyset.
\]
Therefore $\Facnei_\FCC(s)\neq\emptyset$, and thus $\FACNF$ is $\PAC$-serial.

\item \textbf{Independence.}
Assume that $\FGCGF$ is independent. Fix $s \in \FST$ and disjoint coalitions $\FCC,\FDD\subseteq \FAG$ with $\FCC\cap\FDD=\emptyset$.
Let $X\in\Facnei_\FCC(s)$ and $Y\in\Facnei_\FDD(s)$.
Since $\Facnei_\FCC=\FACEF_\FCC$ and $\Facnei_\FDD=\FACEF_\FDD$, there exist $\sigma_\FCC\in\Fav_\FCC(s)$ and $\sigma_\FDD\in\Fav_\FDD(s)$ such that $\Fout_\FCC(s,\sigma_\FCC)=X$ and $\Fout_\FDD(s,\sigma_\FDD)=Y$.

By independence of $\FGCGF$, we have $\sigma_\FCC\cup\sigma_\FDD \in \Fav_{\FCC\cup\FDD}(s)$.
Let $Z=\Fout_{\FCC\cup\FDD}(s,\sigma_\FCC\cup\sigma_\FDD)$.
Then $Z\in\FACEF_{\FCC\cup\FDD}(s)= \Facnei_{\FCC\cup\FDD}(s)$.
Moreover, since $\sigma_\FCC\subseteq \sigma_\FCC\cup\sigma_\FDD$ and $\sigma_\FDD\subseteq \sigma_\FCC\cup\sigma_\FDD$, by outcome monotonicity, we have
$Z\subseteq X$ and $Z\subseteq Y$, and hence $Z\subseteq X\cap Y$.
Therefore $\FACNF$ is $\PAC$-independent.

\item \textbf{Determinism.}
Assume $\FGCGF$ is deterministic. Fix $s\in\FST$ and $X\in\Facnei_\FAG(s)$.
Then $X\in\FACEF_\FAG(s)$, so there exists $\sigma_\FAG\in\Fav_\FAG(s)$ such that $\Fout_\FAG(s,\sigma_\FAG)=X$.
By determinism of $\FGCGF$, $\Fout_\FAG(s,\sigma_\FAG)$ is a singleton; hence so is $X$.
Therefore $\FACNF$ is $\PAC$-deterministic.

\end{itemize}

\paragraph{We show that $\FACNF$ satisfies the four $\PAC$-representativeness conditions.}

\begin{enumerate}
\item \textbf{Actual triviality of the empty coalition.}
Fix $s\in\FST$.
Note that
\[
\Fout_\emptyset(s,\emptyset)= \bigcup \{\Fout_\FAG (s,\sigma_\FAG ) \mid  \sigma_\FAG  \in \FJA_\FAG \}.
\]
If $\Fout_\emptyset(s,\emptyset)\neq\emptyset$, then by definition of $\FACEF_\emptyset$ we have
\[
\Facnei_\emptyset(s)=\FACEF_\emptyset(s)=\{\Fout_\emptyset(s,\emptyset)\},
\]
so $\Facnei_\emptyset(s)$ is a singleton. If $\Fout_\emptyset(s,\emptyset)=\emptyset$, then $\Fav_\emptyset(s)=\emptyset$ and hence
$\Facnei_\emptyset(s)=\FACEF_\emptyset(s)=\emptyset$.

\item \textbf{Liveness.}
Fix $s\in\FST$ and $\FCC\subseteq\FAG$.
Since
\[
\FACEF_\FCC(s)=\{\Fout_\FCC(s,\sigma_\FCC)\mid \sigma_\FCC\in\Fav_\FCC(s)\}
\]
and each $\sigma_\FCC\in\Fav_\FCC(s)$ satisfies $\Fout_\FCC(s,\sigma_\FCC)\neq\emptyset$, we obtain
$\emptyset\notin \FACEF_\FCC(s)$.
Thus $\emptyset\notin \Facnei_\FCC(s)$.

\item \textbf{Actual power inclusion.}
Fix $s\in\FST$ and coalitions $\FCC\subseteq\FDD\subseteq\FAG$.
Let $X\in\Facnei_\FDD(s)=\FACEF_\FDD(s)$.
Then there exists $\sigma_\FDD\in\Fav_\FDD(s)$ such that $\Fout_\FDD(s,\sigma_\FDD)=X$.
Write $\sigma_\FDD=\sigma_\FCC\cup\sigma_{\FDD\setminus\FCC}$ with $\sigma_\FCC\in \FJA_\FCC$ and $\sigma_{\FDD\setminus\FCC}\in \FJA_{\FDD\setminus\FCC}$.
Since $\sigma_\FDD$ is available for $\FDD$, its restriction $\sigma_\FCC$ is available for $\FCC$, i.e., $\sigma_\FCC\in\Fav_\FCC(s)$.
Moreover,
$
\Fout_\FDD(s,\sigma_\FDD)\subseteq \Fout_\FCC(s,\sigma_\FCC).
$
Let $Y=\Fout_\FCC(s,\sigma_\FCC)$. Then $Y\in\FACEF_\FCC(s)=\Facnei_\FCC(s)$ and $X\subseteq Y$, as required.

\item \textbf{Actual power decomposition.}
Fix $s\in\FST$ and coalitions $\FCC\subseteq\FDD\subseteq\FAG$.
Let $X\in\Facnei_\FCC(s)=\FACEF_\FCC(s)$.
Then there exists $\sigma_\FCC\in\Fav_\FCC(s)$ such that $\Fout_\FCC(s,\sigma_\FCC)=X$.
Let
$
\Theta=\{\sigma_\FAG\in\Fav_\FAG(s)\mid \sigma_\FCC\subseteq\sigma_\FAG\}.
$
By the GCI-condition of $\FGCGF$ and Fact \ref{fact:gcgf-basic-consequences},
\[
X=\Fout_\FCC(s,\sigma_\FCC)=\bigcup\{\Fout_\FAG(s,\sigma_\FAG)\mid \sigma_\FAG\in\Theta\}.
\]
Consider the family
\[
\Delta=\{\Fout_\FDD(s,\sigma_\FAG|_\FDD)\mid \sigma_\FAG\in\Theta\}.
\]
For each $\sigma_\FAG\in\Theta$, we have $\Fout_\FAG(s,\sigma_\FAG)\neq\emptyset$ and $\Fout_\FAG(s,\sigma_\FAG)\subseteq \Fout_\FDD(s,\sigma_\FAG|_\FDD)$, hence $\sigma_\FAG|_\FDD\in\Fav_\FDD(s)$ and therefore $\Fout_\FDD(s,\sigma_\FAG|_\FDD)\in\FACEF_\FDD(s)=\Facnei_\FDD(s)$.
Thus $\Delta\subseteq \Facnei_\FDD(s)$.

It remains to show that $X=\bigcup \Delta$.

\smallskip\noindent\textit{($\subseteq$)} Let $u\in X$. Then $u\in\Fout_\FAG(s,\sigma_\FAG)$ for some $\sigma_\FAG\in\Theta$, hence
$u\in \Fout_\FDD(s,\sigma_\FAG|_\FDD)$ and thus $u\in\bigcup\Delta$.

\smallskip\noindent\textit{($\supseteq$)} Let $u\in\bigcup\Delta$. Then $u\in\Fout_\FDD(s,\sigma_\FAG|_\FDD)$ for some $\sigma_\FAG\in\Theta$.
By the GCI-condition of $\FGCGF$, there exists $\lambda_\FAG \in \FJA_\FAG$ such that
$\sigma_\FAG|_\FDD\subseteq \lambda_\FAG$ and $u\in\Fout_\FAG(s,\lambda_\FAG)$.
Since $\sigma_\FCC\subseteq \sigma_\FAG|_\FDD\subseteq \lambda_\FAG$ and $\lambda_\FAG$ is available, we have $\lambda_\FAG\in\Theta$.
Hence
\[
u\in\bigcup\{\Fout_\FAG(s,\tau_\FAG)\mid \tau_\FAG\in\Theta\}= X.
\]

\end{enumerate}

This completes the proof.

\end{proof}

\begin{theorem}[Actual enoughness theorem]
\label{thm:actual-enoughness}

Assume $\FAG=\{a,b\}$.
Let $\FXX\in\FES$.
Every $\PAC$-representative actual neighborhood $\FXX$-frame $\PAC$-represents some general concurrent game $\FXX$-frame.

\end{theorem}

The proof is somewhat involved; we refer to \cite{chen-representation-2026} for an illustrative example.

\begin{proof}
~
\paragraph{Analysis of the task.}

Let $\FACNF = (\FST, \{\Facnei_\FCC \mid \FCC \subseteq \FAG\})$ be an $\PAC$-representative actual neighborhood $\FXX$-frame.
We aim to construct a general concurrent game $\FXX$-frame
\[
\FGCGF = (\FST, \FAC, \{\Fav_\FCC \mid \FCC \subseteq \FAG\}, \{\Fout_\FCC \mid \FCC \subseteq \FAG\})
\]
such that $\FGCGF$ is $\PAC$-represented by $\FACNF$.

All conditions defining general concurrent game $\FXX$-frames are
statewise: for each source state $s$, they involve only the availability
sets $\Fav_\FCC(s)$ and the outcome maps $\Fout_\FCC(s,\cdot)$, and do
not relate these information to those at a distinct source state $t$.
Consequently, the required availability and outcome information can be
constructed independently at each state. The resulting local
constructions can then be combined by using pairwise disjoint,
state-tagged action sets and taking their union as the global action
set.

Therefore, it suffices to proceed as follows: for each $s \in \FST$, starting from the local neighborhood structure
\[
\Fln \;=\; \{\Facnei_\FCC(s) \mid \FCC \subseteq \FAG\},
\]
we construct a local structure
\[
\Flg \;=\; \bigl(\FAC_s,\{\Fav_\FCC(s)\mid \FCC\subseteq\FAG\},\{\Fout_\FCC(s,\cdot)\mid \FCC\subseteq\FAG\}\bigr)
\]
satisfying the following conditions.

\begin{enumerate}
\item \label{itm:local-component}
$\Flg$ is a local component (called a \Fdefs{local game}) of a general concurrent game frame:

\begin{enumerate}
\item $\FAC_s$ is a set of actions.
\item For every $\FCC \subseteq \FAG$, $\Fav_\FCC(s) \subseteq \FJA_\FCC$.
\item For every $\FCC \subseteq \FAG$ and every $\sigma_\FCC \in \FJA_\FCC$, $\Fout_\FCC(s,\sigma_\FCC) \subseteq \FST$.
\item For every $\FCC \subseteq \FAG$:
\[
\Fav_\FCC(s)=\{\sigma_\FCC\in\FJA_\FCC \mid \Fout_\FCC(s,\sigma_\FCC)\neq\emptyset\}.
\]
\item For every $\FCC \subseteq \FAG$ and every $\sigma_\FCC \in \FJA_\FCC$, the action $\sigma_\FCC$ satisfies the GCI-condition at $s$:
\[
\Fout_\FCC(s,\sigma_\FCC)
=\bigcup \{\Fout_\FAG(s,\sigma_\FAG)\mid \sigma_\FAG\in\Fav_\FAG(s) \text{ and } \sigma_\FCC \subseteq \sigma_\FAG\}.
\]

\end{enumerate}

\item \label{itm:representable-locally}
$\Flg$ is $\PAC$-represented by $\Fln$, i.e., for every $\FCC \subseteq \FAG$,
\[
\{\Fout_\FCC(s,\sigma_\FCC)\mid \sigma_\FCC\in\Fav_\FCC(s)\} \;=\; \Facnei_\FCC(s).
\]

\item \label{itm:XX-properties}
$\Flg$ has the properties corresponding to $\FXX$.

\end{enumerate}

\smallskip
Fix $s \in \FST$. In the sequel, to simplify notation, we drop the explicit mention of $s$. More precisely:

\begin{itemize}
  \item We write $\Facnei_\FCC$ for $\Facnei_\FCC(s)$.
  \item We write $\FAC$ for $\FAC_s$.
  \item We write $\Fav_\FCC$ for $\Fav_\FCC(s)$.
  \item We write $\Fout_\FCC(\sigma_\FCC)$ for $\Fout_\FCC(s,\sigma_\FCC)$.
\end{itemize}

\paragraph{The trivial case.}

Assume $\Facnei_\emptyset=\emptyset$. By Fact~\ref{fact:either all empty or none empty}, we have $\Facnei_\FCC=\emptyset$ for every $\FCC \subseteq \FAG$.
Define a game $\Flg=(\FAC,\{\Fav_\FCC\mid \FCC\subseteq\FAG\},\{\Fout_\FCC\mid \FCC\subseteq\FAG\})$ as follows:

\begin{itemize}
  \item Let $\FAC$ be any nonempty set (e.g., a singleton).
  \item For every $\FCC \subseteq \FAG$, let $\Fav_\FCC=\emptyset$.
  \item For every $\FCC \subseteq \FAG$ and every $\sigma_\FCC\in\FJA_\FCC$, let $\Fout_\FCC(\sigma_\FCC)=\emptyset$.
\end{itemize}

It is straightforward to verify that $\Flg$ satisfies the requirements in the analysis above.

Assume now that $\Facnei_\emptyset\neq\emptyset$.

\paragraph{Construction of a local game.}

By Fact~\ref{fact:either all empty or none empty}, none of $\Facnei_a$, $\Facnei_b$, and $\Facnei_\FAG$ is empty.
Assume
\[
\Facnei_a=\{X_i\mid i\in I\},\qquad
\Facnei_b=\{Y_j\mid j\in J\},\qquad
\Facnei_\FAG=\{Z_k\mid k\in K\},
\]
where $I,J,K$ are index sets.

Initialize $\Flg=(\FAC,\{\Fav_\FCC\mid \FCC\subseteq\FAG\},\{\Fout_\FCC\mid \FCC\subseteq\FAG\})$ by setting $\FAC=\emptyset$ and leaving all $\Fav_\FCC$ and $\Fout_\FCC$ empty. We will add actions, availability, and outcomes step by step until $\Flg$ satisfies conditions~\ref{itm:local-component}--\ref{itm:XX-properties}.

Before the construction, we fix some terminology and conventions, and also give some quick preview of what we are going to do next:

\begin{itemize}
\item By \textbf{individual powers} we mean powers of $a$ or of $b$.

\item Actions in $\FAC$ will have the form $\alpha_{n-X-Z}$ and $\beta_{n-Y-Z}$, where
$n\in\{1,2,3\}$, $X\in\Facnei_a$, $Y\in\Facnei_b$, and $Z\in\Facnei_\FAG$ satisfy $Z\subseteq X$ and $Z\subseteq Y$.
We regard the index $n$ as indicating the \emph{level} (three groups). Joint actions of $\FAG$ will have the form $(\alpha_{n-X-Z},\beta_{m-Y-Z'})$

All $\alpha_{n-X-Z}$ will be available actions of $a$, and all $\beta_{n-Y-Z}$ will be available actions of $b$.
We will stipulate
\[
\Fout_a(\alpha_{n-X-Z})=X
\quad\text{and}\quad
\Fout_b(\beta_{n-Y-Z})=Y.
\]
We say that $\alpha_{n-X-Z}$ is a \Fdefs{name} of $X$, and that $\beta_{n-Y-Z}$ is a name of $Y$.

Available joint actions of $\FAG$ will be introduced gradually.
Not every such pair must be available, since independence need not hold.
Whenever we declare $(\alpha_{n-X-Z}, \ab \beta_{m-Y-Z'})$ to be available, we will simultaneously specify its outcome set
$\Fout_\FAG(\alpha_{n-X-Z},\ab \beta_{m-Y-Z'})=Z''$.
In that case we say that the joint action is a name of $Z''$.

\item For any available action $\alpha$ of $a$, saying that we \emph{make $\alpha$ satisfy the GCI-condition} means that we ensure
\[
\Fout_a(\alpha)
=
\bigcup \left\{ \Fout_\FAG(\alpha,\beta) \;\middle|\; \beta \text{ is available for } b \text{ and } (\alpha,\beta) \text{ is available for } \FAG \right\}.
\]

The same interpretation applies to any available action $\beta$ of $b$.

\item Suppose $\alpha$ is an available action of $a$ and $\beta$ is an available action of $b$. Saying that we \emph{pair $\alpha$ and $\beta$} means that we make $(\alpha,\beta)$ an available action of $\FAG$ and specify its outcome set.

\item 

Let $\alpha$ be an available action of $a$ and $\beta$ an available action of $b$.
Declaring $(\alpha,\beta)$ to be an available joint action with outcome $Z$ is called \Fdefs{safe} (for $\alpha$ and $\beta$) if
\[
Z \subseteq \Fout_a(\alpha)\cap \Fout_b(\beta).
\]
In that case, adding this pair cannot destroy the GCI-condition for $\alpha$ or for $\beta$: if either action already satisfies its GCI-condition before the addition, it still satisfies it afterwards.

\item Any joint action not explicitly declared available will be treated as unavailable, with outcome $\emptyset$.

\end{itemize}

\smallskip

The construction consists of three steps.

\paragraph{Step 1: introducing names for individual powers.}

For every $X\in\Facnei_a$, every $Z\in\Facnei_\FAG$ with $Z\subseteq X$, and every $n \in \{1,2,3\}$, let $\alpha_{n-X-Z}$ be an available action of $a$ such that $\Fout_a(\alpha_{n-X-Z})=X$.

For every $Y\in\Facnei_b$, every $Z\in\Facnei_\FAG$ with $Z\subseteq Y$, and every $n \in \{1,2,3\}$, let $\beta_{n-Y-Z}$ be an available action of $b$ such that $\Fout_b(\beta_{n-Y-Z})=Y$.

The reason that we need three levels of individual actions will be given later.

\paragraph{Step 2: enforcing the GCI-condition for names of individual powers.}

For each $X\in\Facnei_a$ define $\Delta_X=\{Z\in\Facnei_\FAG\mid Z\subseteq X\}$, and for each $Y\in\Facnei_b$ define $\Delta_Y=\{Z\in\Facnei_\FAG\mid Z\subseteq Y\}$. 

\emph{Note that by Fact~\ref{fact:X union Delta X}, $\bigcup \Delta_X = X$ for each $X\in\Facnei_a$ and $\bigcup \Delta_Y = Y$ for each $Y\in\Facnei_b$.}

\begin{enumerate}

\item \textbf{Make all Group~1 names of $a$ satisfy the GCI-condition.}
Pick an available action $\alpha_{1-X-Z}$ of $a$.
For every $Z'\in\Delta_X$, choose some $Y\in\Facnei_b$ such that $Z'\subseteq Y$, and declare $(\alpha_{1-X-Z}, \ab \beta_{2-Y-Z'})$ to be an available action of $\FAG$ with
\[
\Fout_\FAG(\alpha_{1-X-Z},\beta_{2-Y-Z'})=Z'.
\]
\emph{Note that such a $Y$ exists for every $Z'\in\Delta_X$, by actual power inclusion in $\Fln$. Note that for every $Z'\in\Delta_X$, $\beta_{2-Y-Z'}$ is a Group 2 action of $b$.}
\emph{After this, every $Z'\in\Delta_X$ has a name extending $\alpha_{1-X-Z}$, and $\alpha_{1-X-Z}$ satisfies the GCI-condition.}
Repeat this procedure for any remaining Group~1 action of $a$ not yet satisfying the GCI-condition.

\item \textbf{Make all Group~1 names of $b$ satisfy the GCI-condition.}
Proceed analogously for Group~1 actions of $b$ not yet satisfying the GCI-condition.
\emph{Here we only pair such actions of $b$ with Group~2 actions of $a$.}

\item \textbf{Make all Group~2 names of $a$ satisfy the GCI-condition.}
Proceed analogously.
\emph{Here we only pair such actions of $a$ with Group~3 actions of $b$. Note that the pairing is safe.}

\item \textbf{Make all Group~2 names of $b$ satisfy the GCI-condition.}
Proceed analogously.
\emph{Here we only pair such actions of $b$ with Group~3 actions of $a$. Note that the pairing is safe.}

\item \textbf{Make all Group~3 names of $a$ satisfy the GCI-condition.}
Proceed analogously.
\emph{Note that at this point, no Group~1 action of $b$ has been paired with a Group~3 action of $a$.
Here we only pair Group~3 actions of $a$ with Group~1 actions of $b$. Note that the pairing is safe.}

\item \textbf{Make all Group~3 names of $b$ satisfy the GCI-condition.}
Proceed analogously.
\emph{Note that at this point, no Group~1 action of $a$ has been paired with a Group~3 action of $b$.
Here we only pair Group~3 actions of $b$ with Group~1 actions of $a$. Note that the pairing is safe.}

\end{enumerate}

\emph{We now explain why three levels of individual actions are used. The idea behind Step~2 is to make the names of individual powers satisfy the GCI-condition one level at a time. For a name $\alpha$ of an individual power, this means ensuring that its outcome set is exactly the union of the outcome sets of its available joint extensions. To achieve this for names at one level, we add joint extensions using names of the other agent at the next level: Group~1 names are handled using Group~2 names, and Group~2 names are handled using Group~3 names. If this procedure were continued linearly, it would seem to require a further level to handle Group~3 names. The point of the three-level construction is that the process closes cyclically. When Group~3 names are treated, the corresponding Group~1 names of the other agent are already available and already satisfy their own GCI-condition, and they have not yet been paired with the Group~3 names currently under consideration. We can therefore pair Group~3 names with Group~1 names. These new pairings are safe, so they do not destroy the GCI-condition already achieved for the Group~1 names. Thus no fourth level is needed.}

\paragraph{Step 3: pair names of individual powers whenever possible.}

Pick any available actions $\alpha_{n-X-Z'}$ of $a$ and $\beta_{m-Y-Z''}$ of $b$ (at arbitrary levels $n,m\in\{1,2,3\}$) such that:

\begin{itemize}
\item They have not been paired yet, and
\item There exists $Z\in\Facnei_\FAG$ with $Z\subseteq \Fout_a(\alpha_{n-X-Z'})\cap \Fout_b(\beta_{m-Y-Z''})$.
\end{itemize}

Choose such a $Z$, and declare $(\alpha_{n-X-Z'},\beta_{m-Y-Z''})$ to be an available action of $\FAG$ with
\[
\Fout_\FAG(\alpha_{n-X-Z'},\beta_{m-Y-Z''})=Z.
\]
Repeat for all remaining pairs satisfying the conditions above.

\paragraph{Final local game.}

Finally, define $\Flg = (\FAC,\{\Fav_\FCC\mid \FCC\subseteq\FAG\},\{\Fout_\FCC\mid \FCC\subseteq\FAG\})$ as follows:

\begin{itemize}
\item Let $\FAC$ be the set of all actions specified at Step 1.
\item Let $\Fav_\emptyset=\{\emptyset\}$ and $\Fout_\emptyset(\emptyset)=\bigcup\Facnei_\FAG$.
\item Let $\Fav_a,\Fout_a$ be as specified above; similarly for $\Fav_b,\Fout_b$; and let $\Fav_\FAG,\Fout_\FAG$ be as specified above.
\end{itemize}

We verify the ODA-condition and GCI-condition for $\Flg$.

\begin{itemize}

\item

For the ODA-condition, first consider the nontrivial case. Every
individual action declared available in Step~1 has an outcome
$X\in\Facnei_a$ or $Y\in\Facnei_b$, and every grand-coalition action
declared available in Steps~2 and~3 has an outcome
$Z\in\Facnei_\FAG$. All these outcome sets are nonempty by liveness.
Moreover,
\[
\Fout_\emptyset(\emptyset)=\bigcup\Facnei_\FAG\neq\emptyset,
\]
so $\emptyset\in\Fav_\emptyset$. Conversely, every action or joint
action not explicitly declared available is assigned the empty
outcome. Hence, for every $\FCC\subseteq\FAG$,
\[
\Fav_\FCC
=
\{\sigma_\FCC\in\FJA_\FCC
  \mid \Fout_\FCC(\sigma_\FCC)\neq\emptyset\}.
\]
In the trivial case, all availability sets and outcome sets are empty,
so the same equality holds.

\item 

For the GCI-condition, the case $\FCC=\FAG$ is immediate, since a
grand-coalition action has only itself as a grand-coalition extension.
Now let $\alpha_{n-X-Z}\in\Fav_a$. In Step~2, extensions of this action
are introduced whose outcome sets cover every member of
\[
\Delta_X=\{Z'\in\Facnei_\FAG\mid Z'\subseteq X\}.
\]
Since $\bigcup\Delta_X=X$ by Fact~\ref{fact:X union Delta X}, the union
of the outcomes of these extensions is $\Fout_a(\alpha_{n-X-Z})=X$.
Every extension added in Step~3 is safe, and therefore has its outcome
contained in $X$; such extensions do not change this union. The
argument for available actions of $b$ is symmetric. If an individual
action is unavailable, no available grand-coalition action extends it,
and both sides of the GCI equality are empty.

Finally, consider $\FCC=\emptyset$. Every available grand-coalition
action has an outcome in $\Facnei_\FAG$. Conversely, every
$Z\in\Facnei_\FAG$ occurs as the outcome of some available
grand-coalition action: by actual power inclusion, choose
$X\in\Facnei_a$ with $Z\subseteq X$; Step~2 assigns outcome $Z$ to an
extension of $\alpha_{1-X-Z}$. Consequently,
\[
\bigcup_{\sigma_\FAG\in\FJA_\FAG}
\Fout_\FAG(\sigma_\FAG)
=
\bigcup\Facnei_\FAG
=
\Fout_\emptyset(\emptyset).
\]
Thus the GCI-condition holds for every coalition. Hence $\Flg$
satisfies the defining conditions of a local general concurrent game
frame.

\end{itemize}

Hence $\Flg$ is a game.

\paragraph{We show that $\Flg$ is $\PAC$-represented by $\Fln$.}

It suffices to show that for every $\FCC \subseteq \FAG$,
\[
\{\Fout_\FCC(\sigma_\FCC)\mid \sigma_\FCC\in\Fav_\FCC\} \;=\; \Facnei_\FCC.
\]
The inclusion ``$\subseteq$'' is immediate from the construction.
For $\FCC=\emptyset$, by the definition of the final local game, we have the following equalities: 
$
\begin{aligned}
\{\Fout_\emptyset(\sigma_\emptyset)
  \mid \sigma_\emptyset\in\Fav_\emptyset\}
=
\{\Fout_\emptyset(\emptyset)\}
=
\left\{\bigcup\Facnei_\FAG\right\}
=
\Facnei_\emptyset,
\end{aligned}
$
where the last equality follows from
Fact~\ref{fact:empty-coalition-power}.
For $\FCC=\{a\}$ or $\FCC=\{b\}$, the reverse inclusion is also immediate from Step~1.
It remains to show the reverse inclusion for $\FCC=\FAG$, i.e., that every $Z\in\Facnei_\FAG$ has a name.

Let $Z\in\Facnei_\FAG$. By actual power inclusion, $Z\subseteq X$ for some $X\in\Facnei_a$, hence $Z\in\Delta_X$.
Consider the action $\alpha_{1-X-Z}$ introduced in Step~1.
In Step~2 we choose some $Y\in\Facnei_b$ with $Z\subseteq Y$ and declare $(\alpha_{1-X-Z},\beta_{2-Y-Z})$ to be an available action of $\FAG$ with
$\Fout_\FAG(\alpha_{1-X-Z},\beta_{2-Y-Z})=Z$.
Thus $Z$ has a name, as required.
Therefore, $\Flg$ is $\PAC$-represented by $\Fln$.

\paragraph{We show that $\Flg$ has the properties corresponding to $\FXX$.}

\begin{itemize}
\item \textbf{Seriality.}
Assume $\mathtt{S}\in\FXX$. Then $\FACNF$ is $\PAC$-serial, hence $\Facnei_\FCC\neq\emptyset$ for all $\FCC\subseteq\FAG$.
In particular, pick $X\in\Facnei_a$ and choose $Z\in\Facnei_\FAG$ with $Z\subseteq X$.
Then $\alpha_{1-X-Z}$ is available for $a$, so $\Fav_a\neq\emptyset$.
Similarly, $\Fav_b\neq\emptyset$.
Also, since $\Facnei_\FAG\neq\emptyset$, there exists $\sigma_\FAG\in\Fav_\FAG$ with $\Fout_\FAG(\sigma_\FAG)\in\Facnei_\FAG$, hence $\Fav_\FAG\neq\emptyset$.
Finally, $\Fav_\emptyset=\{\emptyset\}$ by definition. Therefore $\Fav_\FCC\neq\emptyset$ for all $\FCC\subseteq\FAG$.

\item \textbf{Independence.}
Assume $\mathtt{I}\in\FXX$. Then $\FACNF$ is $\PAC$-independent. Let $\FCC,\FDD\subseteq\FAG$ with $\FCC\cap\FDD=\emptyset$, and let
$\sigma_\FCC\in\Fav_\FCC$ and $\sigma_\FDD\in\Fav_\FDD$.
We show that $\sigma_\FCC\cup\sigma_\FDD\in\Fav_{\FCC\cup\FDD}$.

If $\FCC=\emptyset$ or $\FDD=\emptyset$, the claim is immediate.
Otherwise, since $\FAG=\{a,b\}$, both $\FCC$ and $\FDD$ are singletons; without loss of generality let $\FCC=\{a\}$ and $\FDD=\{b\}$.
Write $\sigma_a=\alpha$ and $\sigma_b=\beta$, and let $\Fout_a(\alpha)=X$ and $\Fout_b(\beta)=Y$.

By $\PAC$-independence of $\FACNF$, there exists
$Z\in\Facnei_\FAG$ such that $Z\subseteq X\cap Y$.

If $(\alpha,\beta)$ was already declared available in Step~2, then
$(\alpha,\beta)\in\Fav_\FAG$. Otherwise, the existence of such a $Z$ shows
that the unpaired actions $\alpha$ and $\beta$ satisfy the criterion
of Step~3. Hence Step~3 declares $(\alpha,\beta)$ available and
assigns it some outcome $Z'\in\Facnei_\FAG$ with
$Z'\subseteq X\cap Y$. Thus, in either case,
$
\sigma_a\cup\sigma_b
=
(\alpha,\beta)
\in\Fav_\FAG
=
\Fav_{\FCC\cup\FDD}.
$
Therefore $\Flg$ satisfies independence.

\item \textbf{Determinism.}
Assume $\mathtt{D}\in\FXX$. Then $\FACNF$ is $\PAC$-deterministic, so every $Z\in\Facnei_\FAG$ is a singleton.
By construction, for every $\sigma_\FAG\in\Fav_\FAG$ we have $\Fout_\FAG(\sigma_\FAG)\in\Facnei_\FAG$, hence $\Fout_\FAG(\sigma_\FAG)$ is a singleton.
\end{itemize}

\end{proof}

\paragraph{Remark}

The conditions used above remain necessary for arbitrary finite agent
sets, but they cease to be sufficient once three agents are present.
Actual power inclusion and actual power decomposition constrain powers
only along inclusions of coalitions. With three agents, however,
incomparable coalitions may overlap. For example, the restrictions of
a full joint action to $\{a,b\}$ and $\{a,c\}$ must agree on their
common $a$-component; consequently, their outcome sets must both be
contained in the outcome set induced by the same action of $a$. The
conditions used here do not ensure that the witnesses along these two
inclusion chains can be chosen compatibly. This obstruction does not
arise with two agents, since no two incomparable proper nonempty
coalitions have a nonempty intersection. Thus, a general finite-agent
representation theorem requires additional coherence conditions for
the actual powers of overlapping coalitions. A concrete three-agent
counterexample is given in the full paper~\cite{chen-representation-2026}.

\section{Concluding remarks}
\label{section: Concluding remarks}

Concurrent game frames are the most widely used semantic frames for logics of strategic reasoning. Li and Ju~\cite{li-minimal-2025, li-completeness-2026} argued that standard concurrent game frames rely on three assumptions that may be overly strong, namely seriality, independence of agents, and determinism. They therefore proposed eight classes of general concurrent game frames, distinguished by which of these three properties are satisfied.
In this work, under the assumption that there are two agents, we showed that with respect to actual powers, the eight classes of general concurrent game frames are represented by eight corresponding classes of actual neighborhood frames. This generalizes the result in \cite{benthem-new-2019}, which proves a partial representation theorem, for singleton coalitions only and not for the empty or grand coalitions, for the most standard case.

The preceding remark identifies the obstruction to extending the
enoughness direction beyond two agents. A natural next step is to
formulate appropriate general finite-agent coherence conditions and seek a genuine general finite-agent representation theorem
for actual powers.

\section*{Acknowledgments}

We thank the three anonymous reviewers of the 17th International Symposium on Games, Automata, Logics, and Formal Verification for their valuable comments and constructive suggestions.
Fengkui Ju was supported by ``the Fundamental Research Funds
for the Central Universities'' in China (No. 1253200057).
Thomas {\AA}gotnes was supported by Major Project of National Social Science Foundation of China (Grant No. 24\&ZD227).

\bibliographystyle{eptcs}
\bibliography{Strategy-reasoning}

\end{document}